\documentclass[a4paper]{llncs}

\pdfoutput=1

\usepackage{amstext}
\usepackage{amsmath}
\usepackage{amssymb}
\usepackage{url}
\usepackage{enumerate}
\usepackage{graphicx}
\usepackage{picinpar}
\usepackage{alltt}
\usepackage{wrapfig}
\usepackage{paralist}
\usepackage{xspace}
\usepackage{color}
\usepackage{wrapfig}
\usepackage{times}
\usepackage{paralist}

\newcommand{\pdr}{\textsc{IC3}\xspace}
\newcommand{\pdria}{\textsc{IC3+IA}\xspace}

\newcommand{\T}{T\xspace}
\newcommand{\mathsat}{\textsc{MathSAT}\xspace}

\newcommand{\lra}{LRA\xspace}
\newcommand{\lia}{LIA\xspace}
\newcommand{\bv}{BV\xspace}

\newcommand{\concpdr}{\textsc{IC3(\lra)}\xspace}
\newcommand{\conctreepdr}{\textsc{TreeIC3+ITP(\lra)}\xspace}
\newcommand{\abspdr}[1]{{\pdria}(#1)\xspace}
\newcommand{\abstreepdr}[1]{\textsc{TreeIC3+IA(#1)}\xspace}
\newcommand{\zthree}{\textsc{z3}\xspace}
\newcommand{\abcdprove}{\textsc{ABC-dprove}\xspace}
\newcommand{\abcpdr}{\textsc{ABC-pdr}\xspace}
\newcommand{\tip}{\textsc{Tip}\xspace}
\newcommand{\icthree}{\textsc{IC3ref}\xspace}
\newcommand{\pkind}{\textsc{pKind}\xspace}
\newcommand{\kind}{\textsc{Kind}\xspace}
\newcommand{\atmoc}{\textsc{Atmoc}\xspace}

\newenvironment{small2}{\fontsize{8}{10}\selectfont}{\normalsize}
\newsavebox{\fmbox}
\newenvironment{fmpage}[1]
{\begin{lrbox}{\fmbox}\begin{minipage}{#1}}
{\end{minipage}\end{lrbox}\fbox{\usebox{\fmbox}}}

\newcommand{\mktuple}[1]{\ensuremath{\langle{#1}\rangle}}
\newcommand{\vars}{X} 

\newcommand{\evars}{\overline{\vars}} 
\newcommand{\abs}[1]{\widehat{#1}} 
\newcommand{\absrel}{\textit{H}}
\newcommand{\target}{\ensuremath{{\neg P}}}
\newcommand{\fset}[1]{\mathbb{#1}} 
\newcommand{\pred}{\ensuremath{p}}
\newcommand{\preds}{\ensuremath{\fset{P}}}
\newcommand{\abspath}[1][k]{\widehat{\textit{Path}}_{#1,\preds}}
\newcommand{\pequiv}[1][\preds]{\textit{EQ}_{#1}}
\newcommand{\absf}[2]{\ensuremath{\abs{#2}_{#1}}}
\newcommand{\absfs}[2]{\ensuremath{\abs{#2}}}
\newcommand{\bmc}[2]{\ensuremath{\text{BMC}^{#1}_{#2}}}

\newcommand{\relind}{\ensuremath{RelInd}}
\newcommand{\absrelind}{\ensuremath{AbsRelInd}}

\newcommand{\true}{\textit{true}\xspace}
\newcommand{\false}{\textit{false}\xspace}

\title{IC3 Modulo Theories via\\ Implicit Predicate Abstraction}

\author{\ 
Alessandro Cimatti \and Alberto Griggio \and Sergio Mover \and Stefano Tonetta
}

\urldef{\emailauth}\path|{cimatti, griggio, mover, tonettas}@fbk.eu|
\institute{Fondazione Bruno Kessler\\\emailauth}

\begin{document}

\maketitle

\begin{abstract}
We present a novel approach for generalizing the IC3 algorithm for
invariant checking from finite-state to
infinite-state transition systems, expressed over some background
theories.
The procedure is based on a tight integration of IC3 with Implicit (predicate)
Abstraction, a technique that expresses abstract
transitions without computing explicitly the abstract system and is
incremental with respect to the addition of predicates. 
In this scenario, IC3 operates only at the Boolean level of the abstract state
space, discovering inductive clauses over the abstraction predicates.
Theory reasoning is confined within the underlying SMT solver,
and applied transparently when performing satisfiability checks.
When the current abstraction allows for a spurious counterexample, it
is refined by discovering and adding a sufficient set of new predicates.  
Importantly, this can be done in a completely incremental manner,
without discarding the clauses found in the previous search.

The proposed approach has two key advantages.
First, unlike current SMT generalizations of IC3, it allows to handle
a wide range of background theories without relying on ad-hoc
extensions, such as quantifier elimination or theory-specific clause
generalization procedures, which might not always be available, and
can moreover be inefficient. 
Second, compared to a direct exploration of the concrete transition
system, the use of abstraction gives a significant performance
improvement, as our experiments demonstrate.
\end{abstract}

\section{Introduction}
\label{sec:intro}

IC3~\cite{bradley} is an algorithm for the verification of invariant
properties of transition systems. It builds an over-approximation of
the reachable state space, using clauses obtained by generalization
while disproving candidate counterexamples.
In the case of finite-state systems, the algorithm is implemented on
top of Boolean SAT solvers, fully leveraging their features.  IC3 has
demonstrated extremely effective, and it is a fundamental core in all
the engines in hardware verification.

There have been several attempts to lift IC3 to the case of
infinite-state systems, for its potential applications to software,
RTL models, timed and hybrid systems, although the problem is in
general undecidable. These approaches are set in the framework of
Satisfiability Modulo Theory (SMT)~\cite{smtoverview} and hereafter
are referred to as IC3 Modulo
Theories~\cite{CGCAV12,DBLP:conf/formats/KindermannJN12,HoderB12,DBLP:conf/date/WelpK13}:
the infinite-state transition system is symbolically described by
means of SMT formulas, and an SMT solver plays the same role of
the SAT solver in the discrete case. The key difference is the need in
IC3 Modulo Theories for specific theory reasoning to deal with candidate
counterexamples. This led to the development of various techniques,
based on quantifier elimination or theory-specific clause
generalization procedures. Unfortunately, such extensions are
typically ad-hoc, and might not always be applicable in all theories
of interest. Furthermore, being based on the fully detailed SMT
representation of the transition systems, some of these solutions
(e.g. based on quantifier elimination) can be highly inefficient.

We present a novel approach to IC3 Modulo Theories, which is able to
deal with infinite-state systems by means of a tight integration with
\emph{predicate abstraction} (PA)~\cite{GS97}, a standard abstraction technique that
partitions the state space according to the equivalence relation induced by a set of
predicates.
In this work, we leverage \emph{Implicit Abstraction}
(IA)~\cite{fm09}, which allows to express abstract transitions without
computing explicitly the abstract system, and is fully incremental
with respect to the addition of new predicates.
In the resulting algorithm, called \pdria, the search proceeds as if
carried out in an abstract system induced by the set of current
predicates $\preds$ -- in fact, \pdria only generates clauses over
\preds. The key insight is to exploit IA to obtain an abstract version
of the relative induction check.
%
%
When an abstract counterexample is found, as in Counter-Example Guided
Abstraction-Refinement (CEGAR), it is simulated in the concrete space
and, if spurious, the current abstraction is refined by adding a set
of predicates sufficient to rule it out.

%
The proposed approach has several advantages.
First, unlike current SMT generalizations of IC3, \pdria allows to
handle a wide range of background theories without relying on ad-hoc
extensions, such as quantifier elimination or theory-specific clause
generalization procedures.
The only requirement is the availability of an effective technique for
abstraction refinement, for which various solutions exist for many
important theories (e.g. interpolation~\cite{HJMM04}, unsat core
extraction, or weakest precondition).
Second, the analysis of the infinite-state transition system is now
carried out in the abstract space, which is often as effective as an
exact analysis, but also much faster.  Finally, the approach is
completely incremental, without having to discard or reconstruct
clauses found in the previous iterations.

We experimentally evaluated \pdria on a set of benchmarks from
heterogeneous sources~\cite{swmcc,kind-fmcad08,DBLP:conf/formats/KindermannJN12},
with very positive results.
First,
our implementation of \pdria is significantly more expressive than the SMT-based IC3 of \cite{CGCAV12},
being able to handle not only the theory of Linear Rational Arithmetic (\lra) like \cite{CGCAV12},
but also those of Linear Integer Arithmetic (\lia) and fixed-size bit-vectors (\bv).
Second, in terms of performance \pdria proved to be uniformly superior to a wide range of 
alternative techniques and tools,
including state-of-the-art implementations of the bit-level IC3 algorithm (\cite{fmcad-een,tip,ic3ref}),
other approaches for IC3 Modulo Theories (\cite{CGCAV12,HoderB12,DBLP:conf/formats/KindermannJN12}),
and techniques based on k-induction and invariant discovery (\cite{kind-fmcad08,pkind}).
%
%
A remarkable property of \pdria is that it can deal with a large
number of predicates: in several benchmarks, \emph{hundreds of
  predicates} were introduced during the search. Considering that an
explicit computation of the abstract transition relation (e.g. based
on All-SMT~\cite{LahiriNO06}) often becomes impractical with a few
dozen predicates, we conclude that IA is fundamental to scalability,
allowing for efficient reasoning in a fine-grained abstract space.

%
The rest of the paper is structured as follows.
In Section~\ref{sec:background} we present some background on IC3 and
Implicit Abstraction.
In Section~\ref{sec:implabsic3} we describe \pdria and prove its formal
properties.
In Section~\ref{sec:related-work} we discuss the related work.
In Section~\ref{sec:expeval} we experimentally evaluate our method.
In Section~\ref{sec:conclusion} we draw some conclusions and present
directions for future work.

\section{Background}
\label{sec:background}

\subsection{Transition Systems}
\label{sec-background-trans}

Our setting is standard first order logic. 
We use the standard notions of theory, satisfiability, validity, logical consequence. 
We denote formulas with $\varphi, \psi, I, T, P$, variables with $x$, $y$,
and sets of variables with $\vars$, $Y$, $\evars$, $\abs{\vars}$.
Unless otherwise specified, we work on quantifier-free formulas,
and we refer to 0-arity predicates as Boolean variables, and to 0-arity uninterpreted functions as (theory) variables.
A literal is an atom or its negation. A \emph{clause} is a disjunction of literals,
whereas a \emph{cube} is a conjunction of literals. 
If $s$ is a cube $l_1 \land \ldots \land l_n$, with $\neg s$ we denote the clause $\neg l_1 \lor \ldots \lor \neg l_n$, and vice versa.
A formula is in conjunctive normal form (CNF) if it is a conjunction of clauses,
and in disjunctive normal form (DNF) if it is a disjunction of cubes.
With a little abuse of notation, 
we might sometimes denote formulas in CNF
$C_1 \land \ldots \land C_n$ as sets of clauses $\{C_1, \ldots, C_n\}$, and vice versa. 
If $X_1, \ldots, X_n$ are a sets of variables and $\varphi$ is a formula,
we might write $\varphi(X_1, \ldots, X_n)$ to indicate that all the variables occurring in $\varphi$ are elements of $\bigcup_i X_i$.
For each variable $x$, we assume that there exists a corresponding variable $x'$ (the \emph{primed version} of $x$).
If $\vars$ is a set of variables, $\vars'$ is the set obtained by
replacing each element $x$ with its primed version ($\vars' = \{x'
\mid x \in \vars\}$), $\evars$ is the set obtained by replacing
each $x$ with $\overline{x}$ ($\overline{\vars} = \{\overline{x} \mid x \in
\vars\}$) and $\vars^n$ is the set obtained by adding $n$ primes to 
each variable ($\vars^n = \{x^n \mid x \in \vars\}$).

Given a formula $\varphi$, $\varphi'$ is the formula obtained by
adding a prime to each variable occurring in $\varphi$.
Given a theory \T, we write $\varphi \models_\T \psi$ (or simply
$\varphi \models \psi$) to denote that the formula $\psi$ is a logical
consequence of $\varphi$ in the theory \T.

A \emph{transition system} $S$ is a tuple $S=\mktuple{X,I,T}$ where
$X$ is a set of (state) variables, $I(X)$ is a formula representing
the initial states, and $T(X,X')$ is a formula representing the
transitions.
A \emph{state} of $S$ is an assignment to the variables
$X$. A \emph{path} of $S$ is a finite sequence $s_0,s_1,\ldots,s_k$ of
states such that $s_0\models I$ and for all $i$, $0\leq i<k$,
$s_i,s'_{i+1}\models T$. 

Given a formula $P(X)$, the \emph{verification problem} denoted with
$S\models P$ is the problem to check if for all paths
$s_0,s_1,\ldots,s_k$ of $S$, for all $i$, $0\leq i\leq k$, $s_i\models
P$. Its dual is the \emph{reachability problem}, which is the
problem to find a path $s_0,s_1,\ldots,s_k$ of $S$ such that
$s_k\models \neg P$. $P(X)$ represents the ``good'' states, while
$\neg P$ represents the ``bad'' states.

Inductive invariants are central to solve the verification
problem. $P$ is an inductive invariant iff \begin{inparaenum}[(i)]
\item $I(X) \models P(X)$; and 
\item $P(X) \land T(X, X') \models P(X')$.
\end{inparaenum}
A weaker notion is given by relative inductive invariants: given a
formula $\phi(X)$, $P$ is inductive relative to $\phi$
iff \begin{inparaenum}[(i)]
\item $I(X) \models P(X)$; and 
\item $\phi(X)\land P(X) \land T(X, X') \models P(X')$.
\end{inparaenum}

\subsection{\pdr with SMT}

\pdr~\cite{bradley} is an efficient algorithm for the verification of
finite-state systems, with Boolean state variables and propositional
logic formulas.  \pdr was subsequently extended to the SMT case in
\cite{CGCAV12,HoderB12}.  In the following, we present its main ideas,
following the description of \cite{CGCAV12}.  For brevity, we have to
omit several important details, for which we refer to
\cite{bradley,CGCAV12,HoderB12}.

Let $S$ and $P$ be a transition system and a set of good states as in \S\ref{sec-background-trans}.
The \pdr algorithm tries to prove that $S\models P$ by finding
a formula $F(X)$ such that:
\begin{inparaenum}[(i)]
\item $I(X) \models F(X)$;
\item $F(X) \land T(X, X') \models F(X')$; and
\item $F(X) \models P(X)$.
\end{inparaenum}

In order to construct an inductive invariant $F$, \pdr
maintains a sequence of formulas (called \emph{trace})
$F_0(X), \ldots, F_k(X)$ such that:
\begin{inparaenum}[(i)]
\item $F_0 = I$;
\item $F_i \models F_{i+1}$;
\item $F_i(X) \land T(X, X') \models F_{i+1}(X')$;
\item for all $i<k$, $F_i \models P$.
\end{inparaenum}
Therefore, each element of the trace $F_{i+1}$, called \emph{frame},
is inductive relative to previous one, $F_i$. IC3 strengthens the
frames by finding new relative inductive clauses by checking the
unsatisfiability of the formula:
\begin{equation}\label{eq:ic3-inductive-check}
\relind(F,T,c):=F\land c \land T \land \neg c'.
\end{equation}

More specifically, the algorithm proceeds incrementally, by
alternating two phases: a blocking phase, and a propagation phase.
In the \emph{blocking} phase, the trace is analyzed to prove that
no intersection between $F_k$ and $\neg P(X)$ is possible. If such
intersection cannot be disproved on the current trace, the property
is violated and a counterexample can be reconstructed. During the
blocking phase, the trace is enriched with additional formulas, which can
be seen as strengthening the approximation of the reachable state
space. At the end of the blocking phase, if no violation is found, $F_k\models P$.

The \emph{propagation} phase tries to extend the trace with a new
formula $F_{k+1}$, moving forward the clauses from preceding
$F_i$'s. If, during this process, two consecutive frames
become identical (i.e. $F_i=F_{i+1}$), then a fixpoint is reached, and \pdr 
terminates with $F_i$ being an inductive invariant proving the
property.

In the \textit{blocking} phase \pdr maintains a set of
pairs $(s, i)$, where $s$ is 
a set of states that
can lead to a bad state, and $i > 0$ is a position in the current
trace. 
New formulas (in the form of clauses)
to be added to 
the current trace
are derived by (recursively) proving that a set $s$ of a pair $(s, i)$
is unreachable starting from the formula $F_{i-1}$.
This is done by checking the satisfiability of the formula 
$\relind(F_{i-1},T,\neg{s})$.
If the formula is unsatisfiable, 
then $\neg s$ is \emph{inductive relative to $F_{i-1}$},
and \pdr strengthens $F_i$ by adding $\neg s$ to it\footnote{
  $\neg s$ is actually \emph{generalized} before being added to
  $F_i$. Although this is fundamental for the \pdr effectiveness,
  we do not discuss it for simplicity.}, thus
\emph{blocking} the bad state $s$ at $i$. 
If, instead, \eqref{eq:ic3-inductive-check} is satisfiable,
then the overapproximation $F_{i-1}$ is not strong enough to show that $s$
is unreachable.
In this case, let $p$ be 
a subset of the states in $F_{i-1} \land \neg s$ 
such that all the states in $p$ lead to a state in $s'$ in one transition step.
Then, \pdr continues by trying to show that $p$ is not reachable in one step from $F_{i-2}$ 
(that is, it tries to block the pair $(p, i-1)$).
This procedure continues recursively, possibly generating other pairs
to block at earlier points in the trace, until either \pdr generates
a pair $(q, 0)$, meaning that the system does not satisfy the
property, or the trace is eventually strengthened so that the original
pair $(s, i)$ can be blocked.

A key difference between the original Boolean \pdr and its SMT extensions
in \cite{CGCAV12,HoderB12} is in the way sets of states to be blocked or generalized are constructed.
In the blocking phase, when trying to block a pair $(s, i)$,
if the formula \eqref{eq:ic3-inductive-check} is satisfiable,
then a new pair $(p, i-1)$ has to be generated such that $p$ is a cube in the \emph{preimage of $s$ wrt. $T$}.
In the propositional case, $p$ can be obtained from the model $\mu$ of \eqref{eq:ic3-inductive-check} generated by the SAT solver,
by simply dropping the primed variables occurring in $\mu$.
This cannot be done in general in the first-order case, 
where the relationship between the current state variables $X$ and their primed version $X'$ 
is encoded in the theory atoms,
which in general cannot be partitioned into a primed and an unprimed set.
The solution proposed in \cite{CGCAV12} is to compute $p$ by existentially quantifying \eqref{eq:ic3-inductive-check} 
and then applying an \emph{under-approximated} existential elimination algorithm for linear rational arithmetic formulas.
Similarly, in \cite{HoderB12} a theory-aware generalization algorithm
for linear rational arithmetic (based on interpolation) was proposed,
in order to strengthen $\neg s$ before adding it to $F_i$ after having successfully blocked it.

\subsection{Implicit Abstraction}

\subsubsection{Predicate abstraction}

Abstraction~\cite{CGL94} is used to reduce the search space while
preserving the satisfaction of some properties such as invariants. If
$\abs{S}$ is an abstraction of $S$, if a condition is reachable in
$S$, then also its abstract version is reachable in $\abs{S}$.
Thus, if we prove that a set of states is not reachable in
$\abs{S}$, the same can be concluded for the concrete transition
system $S$.

In Predicate Abstraction \cite{GS97}, the
abstract state-space is described with a set of
predicates.
Given a TS $S$, we select a set $\preds$ of predicates, such that each
predicate $\pred\in\preds$ is a formula over the variables $\vars$
that characterizes relevant facts of the system. For every
$\pred\in\preds$, we introduce a new abstract variable $x_\pred$ and define
$\vars_\preds$ as $\{x_\pred\}_{\pred\in\preds}$. 
The abstraction relation $H_\preds$ is defined as
$\absrel_\preds(\vars,\vars_\preds):=\bigwedge_{\pred\in
\preds}x_\pred\leftrightarrow \pred(\vars)$.
Given a formula $\phi(\vars)$, the abstract version
$\absf{\preds}{\phi}$ is obtained by existentially quantifying the
variables $X$, i.e., $\absf{\preds}{\phi}=\exists
\vars.(\phi(\vars)\wedge \absrel_\preds(\vars,\vars_\preds))$. Similarly for a
formula over $X$ and $X'$, $\absf{\preds}{\phi}=\exists
\vars,\vars'.(\phi(\vars,\vars')\wedge
\absrel_\preds(\vars,\vars_\preds)\wedge \absrel_\preds(\vars',\vars_\preds'))$.  The
abstract system with $\abs{S}_\preds=\mktuple{\vars_\preds,
  \abs{I}_\preds, \abs{T}_\preds}$ is obtained by abstracting the
initial and the transition conditions. In the following, when clear
from the context, we write just $\absf{}{\phi}$ instead of
$\absf{\preds}{\phi}$.

Since most model checkers deal only with quantifier-free formulas, the
computation of $\abs{S}_\preds$ requires the elimination of the existential
quantifiers. This may result in a bottleneck and some techniques compute
weaker/more abstract systems (cfr., e.g., \cite{STT09}).

\subsubsection{Implicit predicate abstraction}
\begin{wrapfigure}{r}{0.5\textwidth}
\vspace{-10mm}
\begin{center}
\scalebox{0.7}{\input{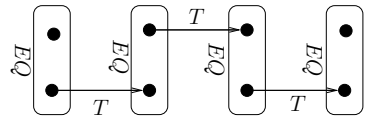}}
\vspace{-5mm}
\end{center}
\caption{Abstract path.}\label{fig:abspath}
\end{wrapfigure}
Implicit predicate abstraction~\cite{fm09} embeds the definition of
the predicate abstraction in the encoding of the path.
This is based on the following formula:
\vspace{-2mm}
\begin{equation}
\pequiv(\vars,\evars):=\bigwedge_{\pred\in\preds} \pred(\vars)
\leftrightarrow \pred(\evars)
\vspace{-2mm}
\end{equation}
\noindent which relate two concrete states corresponding to the same
abstract state. The formula $\abspath:=\bigwedge_{1\leq h<
  k} (T(\evars^{h-1},\vars^h) \wedge \pequiv(\vars^h,\evars^h)) \wedge
T(\evars^{k-1},\vars^{k})$ is satisfiable iff there exists a path of
$k$ steps in the abstract state space.
Intuitively, instead of having a contiguous sequence of transitions,
the encoding represents a sequence of disconnected transitions where
every gap between two transitions is forced to lay in the same abstract
state (see Fig.~\ref{fig:abspath}).
$\bmc{k}{\preds}$ encodes the abstract bounded model checking problem
and is obtained from $\abspath$ by adding the abstract initial and
target conditions:
$\bmc{k}{\preds}=I(\vars^0)\wedge \pequiv(\vars^0,\evars^0)\wedge
\abspath \wedge \pequiv(\vars^k,\evars^k) \wedge \target(\evars^k)$.

\section{IC3 with Implicit Abstraction}
\label{sec:implabsic3}

\subsection{Main idea}

The main idea of \pdria is to mimic how IC3 would work on the abstract
state space defined by a set of predicates $\preds$, but using IA to
avoid quantifier elimination to compute the abstract transition
relation. Therefore, clauses, frames and cubes are restricted to have
predicates in $\preds$ as atoms. We call these clauses, frames and
cubes respectively $\preds$-clauses, $\preds$-formulas, and
$\preds$-cubes. Note that for any $\preds$-formula $\phi$ (and thus
also for $\preds$-cubes and $\preds$-clauses),
$\absfs{\preds}{\phi}=\phi[\vars_\preds/\preds]\wedge\exists
X.(\bigwedge_{p\in\preds}x_p\leftrightarrow p(X))$, and thus
$\absfs{\preds}{(\neg\phi)}=\neg\absfs{\preds}{(\phi)}$. 

The key point of \pdria is to use an abstract version
of the check (\ref{eq:ic3-inductive-check}) to prove that an abstract
clause $\absfs{\preds}{c}$ is inductive relative to the abstract frame
$\absfs{\preds}{F}$:
\begin{eqnarray}\label{eq:ic3ia-inductive-check}
  \absrelind(F,T,c,\preds)&:=&F(\vars) \land c(\vars) \land
  \nonumber\\ && \pequiv[\preds](\vars,\evars) \land T(\evars,\evars')
  \land \pequiv[\preds](\evars',\vars') \land \neg c(\vars')
\end{eqnarray}

{
\sloppypar
\begin{theorem}
Consider a set $\preds$ of predicates, $\preds$-formulas $F$ and a
$\preds$-clause $c$.
$\relind(\absfs{\preds}{F},\absfs{\preds}{T},\absfs{\preds}{c})$ is
satisfiable iff $\absrelind(F,T,c,\preds)$ is satisfiable. In
particular, if $s\models\absrelind(F,T,c,\preds)$, then
$\absfs{\preds}{s}\models\relind(\absfs{\preds}{F},\absfs{\preds}{T},\absfs{\preds}{c})$.
\end{theorem}
}

\begin{proof}
Suppose $s\models\absrelind(F,T,c,\preds)$.  Let us denote with
$\overline{t}$ and ${t}$ the projections of $s$ respectively over
$\evars\cup\evars'$ and over $\vars\cup\vars'$. Then
$\overline{t}\models T$ and therefore
$\absfs{\preds}{\overline{t}}\models\absfs{\preds}{T}$. Since
$s\models\pequiv[\preds](\vars,\evars)\wedge
\pequiv[\preds](\evars',\vars')$, $\absfs{\preds}{t}$ and
$\absfs{\preds}{\overline{t}}$ are the same abstract transition and
therefore $\absfs{\preds}{{t}}\models\absfs{\preds}{T}$. Since
$t\models F\wedge c$, then $\absfs{\preds}{t}\models
\absfs{\preds}{F}\wedge\absfs{\preds}{c}$. Since $t\models \neg c'$,
then $\absfs{\preds}{t}\models \absfs{\preds}{(\neg c')}$ and since
$c$ is a Boolean combinations of $\preds$, then
$\absfs{\preds}{t}\models \neg \absfs{\preds}{c}'$. Thus,
$\abs{s}\models\abs{t}\models\relind(\absfs{\preds}{F},\absfs{\preds}{T},\absfs{\preds}{c})$.

For the other direction, suppose
$\overline{t}\models\relind(\absfs{\preds}{F},\absfs{\preds}{T},\absfs{\preds}{c})$. Then
there exists an assignment $t$ to $\vars\cup\vars'$ such that
$t\models T$ and $\absfs{\preds}{t}=\overline{t}$. Therefore, $t\models F
(\vars) \land c(\vars) \land \pequiv[\preds](\vars,\vars) \land
T(\vars,\vars') \land \pequiv[\preds](\vars',\vars') \land \neg
c(\vars')$, which concludes the proof.
\end{proof}

\subsection{The algorithm}

The \pdria algorithm is shown in Figure~\ref{fig:ic3ia-pseudocode}.
The \pdria has the same structure of \pdr as described in
\cite{fmcad-een}. Additionally, it keeps a set of predicates $\preds$,
which are used to compute new clauses. The only points where \pdria
differs from \pdr (shown in red in Fig.~\ref{fig:ic3ia-pseudocode})
are in picking $\preds$-cubes instead of concrete
states, the use of $\absrelind$ instead of $\relind$, and in the fact
that a spurious counterexample may be found and, in that case, new
predicates must be added.

More specifically, the algorithm consists of a big loop, in which each
iteration is divided into phases, the blocking and the propagation
phases. The blocking phase starts by picking a cube $c$ of predicates
representing an abstract state in the last frame violating the
property. This is recursively blocked along the trace by checking
if$\absrelind(F_{i-1},T,\neg c,\preds)$ is satisfiable.  If the
relative induction check succeeds, $F_i$ is strengthened with a
generalization of $\neg c$.  If the check fails, the recursive
blocking continues with an \emph{abstract predecessor} of $c$, that
is, a $\preds$-cube in $F_i\wedge\neg c$ that leads to $c$ in one
step.  This recursive blocking results in either strengthening of the
trace or in the generation of an \emph{abstract counterexample}.  If
the counterexample can be simulated on the concrete transition system,
then the algorithm terminates with a violation of the
property. Otherwise, it refines the abstraction, adding new predicates
to $\preds$ so that the abstract counterexample is no more a path of
the abstract system.
In the propagation phase, $\preds$-clauses of a
frame $F_i$ that are inductive relative to $F_i$ using
$\absfs{\preds}{T}$ are propagated to the following frame $F_{i+1}$.
As for \pdr, if two consecutive frames are equal, we can conclude that
the property is satisfied by the abstract transition system, and
therefore also by the concrete one.

\begin{figure}[t]
  \centering
    \begin{fmpage}{1.0\linewidth}
  \begin{small2}
    \newcounter{pseudocodecounter}
    \newcommand{\pcl}{%
      \refstepcounter{pseudocodecounter}
      {\scriptsize \arabic{pseudocodecounter}.}\hspace{1em}%
    }
    \newcommand{\pcll}[1]{\pcl\label{#1}} %
    \newcommand{\pcc}[1]{%
      {\it ~\# #1}%
    }
    \newcommand{\pcs}{%
      {\phantom {\scriptsize \arabic{pseudocodecounter}.}\hspace{1em}}
    }
    \newcommand{\absline}[1]{\textcolor{red}{#1}}
    \setcounter{pseudocodecounter}{0}
      \begin{tabbing}
        {\bf bool} \pdria($I$, $T$, $P$, $\preds$):\\
        \pcs xxx \= xxx \= xxx \= xxx \= xxx \= xxx \= \kill
        \pcl \absline{$\preds$ = $\preds \cup \{p \mid p$ is a predicate in $I$ or in $P\}$}\\
        \pcl trace = [$I$] \pcc{first elem of trace is init formula}\\
        \pcl trace.push() \pcc{add a new frame to the trace}\\
        \pcll{ic3-mainloop} {\bf while} True:\\
        \> \pcc{blocking phase}\\
        \pcll{ic3-blocking-begin} \> \absline{{\bf while} there exists a $\preds$-cube $c$ s.t. 
        $c \models \text{trace.last()} \land \neg P$:}\\        
        \pcl \> \> {\bf if not} recBlock($c, \text{trace.size()}-1$):\\
        \> \> \> \pcc{a pair $(s_0, 0)$ is generated} \\
        \pcl \> \> \> \absline{{\bf if} the simulation of $\pi=(s_0,0); \ldots; (s_k, k)$ fails:} \\
        \pcll{ic3-newpred} \> \> \> \> \absline{$\preds := \preds \cup \text{refine}(I, T, P, \preds, \pi)$}\\ 
        \pcll{ic3-blocking-end}\> \> \> {\bf else return} False \pcc{counterexample found}\\[0.5em]
        \> \pcc{propagation phase}\\
        \pcll{ic3-propagation-begin} \> trace.push()\\
        \pcl \> {\bf for} $i = 1$ {\bf to} $\text{trace.size()}-1$:\\
        \pcl \> \> {\bf for each} clause $c \in \text{trace[i]}$:\\
        \pcl \> \> \> \absline{{\bf if} $\absrelind(\text{trace[i]},T,c,\preds) \models \bot$:}\\
        \pcl \> \> \> \> add $c$ to trace[i+1] \\        
        \pcll{ic3-propagation-end} \> \> {\bf if} trace[i] == trace[i+1]: {\bf return} True \pcc{property proved}
      \end{tabbing}
    \setcounter{pseudocodecounter}{0}
      \begin{tabbing}
        \pcc{simplified recursive description, in practice based on priority queue \cite{bradley,fmcad-een}}\\
        {\bf bool} recBlock($s,i$):\\
        \pcs xxx \= xxx \= xxx \= xxx \= xxx \= xxx \= \kill
        \pcl {\bf if} $i == 0$: {\bf return} False \pcc{reached initial states}\\
        \pcl \absline{{\bf while} $\absrelind(\text{trace[i-1]},T,\neg s,\preds) \not\models \bot$:}\\
        \pcl \> \absline{extract a $\preds$-cube $c$ from the Boolean model of $\absrelind(\text{trace[i-1]},T,\neg s,\preds)$} \\
        \> \pcc{$c$ is an (abstract) predecessor of $s$} \\
        \pcl \> {\bf if not} recBlock($c,i-1$): {\bf return} False\\
        \pcl $g$ = generalize($\neg s$, i) \pcc{standard IC3 generalization \cite{bradley,fmcad-een} (using \absrelind)}\\
        \pcl add $g$ to trace[i]\\
        \pcl {\bf return} True\\
      \end{tabbing}
  \end{small2}
    \end{fmpage}
  \caption{High-level description of \pdria (with changes wrt. the Boolean IC3 in red).
    \label{fig:ic3ia-pseudocode}}
\end{figure}

\subsection{Simulation and refinement}
During the search the procedure may find a counterexample in the
abstract space. As usual in the CEGAR framework, we simulate the
counterexample in the concrete system to either find a real
counterexample or to refine the abstraction, adding new predicates to
$\preds$.
Technically, \pdria finds a set of counterexamples $\pi=(s_0,0);
\ldots; (s_k, k)$ instead of a single counterexample, as described
in~\cite{CGCAV12} (i.e. this behaviour depends on the
generalization of a cube performed by ternary simulation or don't care detection).
We simulate $\pi$ as usual via bounded model
checking. Formally, we encode all the paths of $S$ up to $k$ steps
restricted to $\pi$ with: $I(\vars^0) \wedge \bigwedge_{i < k}{T(\vars^i,
    \vars^{i+1})} \wedge P(\vars^k) \wedge \bigwedge_{i \le
    k}{s_k(\vars^k)}$.
If the formula is satisfiable, then there exists a concrete
counterexample that witnesses $S \not\models P$, otherwise $\pi$ is
spurious and we refine the abstraction adding new predicates.
The $\textit{refine}(I,T,\preds,\pi)$ procedure is orthogonal to
\pdria, and can be carried out with several techniques, like
interpolation, unsat core extraction or weakest
precondition, for which there is a wide literature. The only
requirement of the refinement is to remove the spurious
counterexamples $\pi$.
In our implementation we used interpolation to discover predicates,
similarly to~\cite{HJMM04}.

Also, note that in our approach the set of predicates increases
monotonically after a refinement (i.e. we always add new predicates to
the existing set of predicates). Thus, the transition relation is
monotonically strengthened (i.e. since $\preds \subseteq \preds'$,
$\absfs{\preds'}{T}_{\preds'} \rightarrow
\absfs{\preds}{T}_\preds$).
This allows us to \textit{keep all the clauses} in the \pdria frames
after a refinement, enabling a fully incremental approach.

\subsection{Correctness}

\begin{lemma}[Invariants]\label{lem-invariants}
  The following conditions are invariants of \pdria:
  \begin{compactenum}    
    \item
      \label{linv-itm:c1}
      ${F_0} = {I}$;
    \item
      \label{linv-itm:c2}
      for all $i<k$, ${F_i} \models {F}_{i+1}$;
    \item
      \label{linv-itm:c3}
      for all $i<k$, $F_i(\vars) \land \pequiv[\preds](\vars,\evars) \land
      T(\evars,\evars') \land \pequiv[\preds](\evars',\vars') \models
      F_{i+1}(\vars')$;
    \item
      \label{linv-itm:c4}
      for all $i<k$, ${F_i} \models {P}$.
  \end{compactenum}
\end{lemma}

\begin{proof}
Condition~\ref{linv-itm:c1} holds, since initially ${F_0} = {I}$, and
${F_0}$ is never changed.
We prove that the conditions (\ref{linv-itm:c2}-\ref{linv-itm:c4}) are
loop invariants for the main \pdria loop
(line~\ref{ic3-mainloop}).
The invariant conditions trivially hold when entering the loop.

Then, the invariants are preserved by the inner loop at
line~\ref{ic3-blocking-begin}.
The loop may change the content of a frame ${F}_{i+1}$ adding a new
clause $c$ while recursively blocking a cube $(p,i+1)$. $c$ is added
to ${F}_{i+1}$ if the abstract relative inductive check
$\absrelind(F_{i},T,c,\preds)$ holds. Clearly, this preserves the
conditions \ref{linv-itm:c2}-\ref{linv-itm:c3}. In the loop the set of
predicates $\preds$ may change at line \ref{ic3-newpred}. Note that
the invariant conditions still hold in this case.  In particular,
\ref{linv-itm:c3} holds because if $\preds\subseteq\preds'$, then
$\pequiv[\preds']\models\pequiv[\preds]$.
When the inner loop ends, we are guaranteed that ${F}_{k} \models
{P}_\preds$ holds. Thus, condition \ref{linv-itm:c4} is preserved when a
new frame is added to the abstraction in
line~\ref{ic3-propagation-begin}.
Finally, the propagation phase clearly maintains all the invariants
(\ref{linv-itm:c2}-\ref{linv-itm:c4}), by the definition of abstract
relative induction $\absrelind(F_{i},T,c,\preds')$.
\end{proof}

\begin{lemma}
  \label{lemma:soundforabs}
  If \pdria($I$, $T$, $P$, $\preds$) returns
  \true, then $\abs{S_\preds} \models \abs{P_\preds}$.
\end{lemma}

\begin{proof}
  The invariant conditions of the \pdr algorithm hold for the abstract frames:
1) $\absfs{\preds}{F_0} = \absfs{\preds}{I}$;
for all $i<k$, 2) $\absfs{\preds}{F_i} \models \absfs{\preds}{F_{i+1}}$;
3) $\absfs{\preds}{F_i} \land \absfs{\preds}{T} \models \absfs{\preds}{F'_{i+1}}$;
and 4) $\absfs{\preds}{F_i} \models \absfs{\preds}{P}$.

Conditions~1), 2), and 4) follow from
Lemma \ref{lem-invariants}, since $I$, $P$, and $F_i$ are
$\preds$-cubes. Condition~3) follows from Lemma
\ref{lem-invariants}, since $\absfs{\preds}{T}=\exists\evars,\evars'.\pequiv[\preds](\vars,\evars) \land
T(\evars,\evars') \land \pequiv[\preds](\evars',\vars')$ by definition.

By assumption \pdria returns \true and thus $\absfs{\preds}{F_{k-1}} =
\absfs{\preds}{F_k}$. Since the conditions (1-4)
hold, we have that $\absfs{\preds}{F_{k-1}}$ is an inductive invariant
that proves $\absfs{\preds}{S} \models \absfs{\preds}{P}$.
\end{proof}

\begin{theorem}[Soundness]
Let $S=\mktuple{\vars,I,T}$ be a transition system and $P$ a safety
property and $\preds$ be a set of predicates over $\vars$.
The result of \pdria($I$, $T$, $P$, $\preds$) is correct.
\end{theorem}

\begin{proof}
If \pdria($I$, $T$, $P$, $\preds$) returns \true, then
$\abs{S_\preds} \models \abs{P_\preds}$ by
Lemma~\ref{lemma:soundforabs}, and thus $S \models P$.
If \pdria($I$, $T$, $P$, $\preds$) returns \false, then the
simulation of the abstract counterexample in the concrete system
succeeded, and thus $S \not \models P$.
\end{proof}

\begin{lemma}[Abstract counterexample]\label{lem-abscex}
If \pdria finds a counterexample $\pi$, then $\absfs{\preds}{\pi}$ is a
path of $\absfs{\preds}{S}$ violating $\absfs{\preds}{P}$.
\end{lemma}

\begin{proof}
For all $i$, $0\leq i\leq$ trace.size, if $\pi[i]=(s_i,i)$ then $s_i$
is a $\preds$-cube satisfying $F_i$. Moreover, $s_k\models\neg P$. By
Lemma \ref{lem-invariants}, $F_0=I$ and therefore $s_0\models
I$. Since $s_0$, $s_k$, $I$, and $P$ are $\preds$-formulas,
$\absfs{\preds}{s_0}\models \absfs{\preds}{I}$ and
$\absfs{\preds}{s_k}\models \neg\absfs{\preds}{P}$. Again by Lemma
\ref{lem-invariants}, for all $i$, $F_i(\vars) \land
\pequiv[\preds](\vars,\evars) \land T(\evars,\evars') \land
\pequiv[\preds](\evars',\vars') \models F_{i+1}(\vars')$, and thus
$s_i\wedge s_{i+1}'\models \exists
\evars,\evars'.\pequiv[\preds](\vars,\evars) \land T(\evars,\evars')
\land \pequiv[\preds](\evars',\vars')$.  Therefore,
$\absfs{\preds}{s_i}\wedge\absfs{\preds}{s_{i+1}}'\models
\absfs{\preds}{T}$.
\end{proof}

\begin{theorem}[Relative completeness]
Suppose that for some set $\preds$ of predicates,
$\absfs{\preds}{S}\models \absfs{\preds}{P}$.  If, at a certain
iteration of the main loop, \pdria has $\preds$ as set of predicates,
then $\pdria$ returns true.
\end{theorem}

\begin{proof}
Let us consider the case in which, at a certain iteration of the main
loop, $\preds$ is as defined in the premises of theorem. At every
following iteration of the loop, \pdria either finds an abstract counterexample
$\pi$ or strengthens a frame $F_i$ with a new $\preds$-clause. The
first case is not possible, since, by Lemma~\ref{lem-abscex},
$\absfs{\preds}{\pi}$ would be a path of $\absfs{\preds}{S}$ violating
the property. Therefore, at every iteration, \pdria strengthens some
frame with a new $\preds$-clause. Since the number of $\preds$-clauses
is finite and, by Lemma \ref{lem-invariants}, for all $i$, $F_i\models
F_{i+1}$, \pdria will eventually find that $F_i=F_{i+1}$ for some $i$
and return true.
\end{proof}

\section{Related Work}
\label{sec:related-work}

This work combines two lines of research in verification, abstraction
and \pdr.

Among the existing abstraction techniques, predicate
abstraction~\cite{GS97} has been successfully applied to the verification
of infinite-state transition systems, such as
software~\cite{mcmillan}.
Implicit abstraction~\cite{fm09} was first used with k-induction to
avoid the explicit computation of the abstract system. 
In our work, we exploit implicit abstraction in \pdr to avoid
theory-specific generalization techniques, widening the applicability
of \pdr to transition systems expressed over some background
theories. Moreover, we provided the first integration of implicit
abstraction in a CEGAR loop.

The \pdr~\cite{bradley} algorithm has been widely applied to the
hardware domain~\cite{fmcad-een,fmcad-chokler} to prove safety and
also as a backend to prove liveness~\cite{fmcad-liveness}.
In~\cite{VizelGS12}, \pdr is combined with a lazy abstraction technique in the context of hardware verification.
The approach has some similarities with our work, 
but it is limited to Boolean systems, it uses a ``visible variables'' abstraction rather than PA,
and applies a modified concrete version of \pdr for refinement.

Several approaches adapted the original \pdr algorithm to deal with
infinite-state
systems~\cite{CGCAV12,HoderB12,DBLP:conf/formats/KindermannJN12,DBLP:conf/date/WelpK13}.
The techniques presented in \cite{CGCAV12,HoderB12} extend \pdr to
verify systems described in the linear real arithmetic theory.
In contrast to both approaches, we do not rely on theory specific
generalization procedures, which may be expensive, such as quantifier
elimination~\cite{CGCAV12} or may hinder some of the \pdr features,
like generalization (e.g. the interpolant-based generalization
of~\cite{HoderB12} does not exploit relative induction). Moreover,
\pdria searches for a proof in the abstract space.
The approach presented in~\cite{DBLP:conf/formats/KindermannJN12} is
restricted to timed automata since it exploits the finite partitioning
of the region graph. While we could restrict the set of predicates
that we use to regions, our technique is applicable to a much broader
class of systems, and it also allows us to apply conservative
abstractions.
\pdr was also extended to the bit-vector theory
in~\cite{DBLP:conf/date/WelpK13} with an ad-hoc extension, that may
not handle efficiently some bit-vector operators. Instead, our
approach is not specific for bit-vector.

\section{Experimental Evaluation}
\label{sec:expeval}

We have implemented the algorithm described in the previous section
in the SMT extension of IC3 presented in \cite{CGCAV12}.
The tool uses \mathsat~\cite{mathsat5} as backend SMT solver,
and takes as input either a symbolic transition system or 
a system with an explicit control-flow graph (CFG),
in the latter case invoking a specialized ``CFG-aware'' variant of IC3 (TreeIC3, also described in \cite{CGCAV12}).
The discovery of new predicates for abstraction refinement is performed 
using the interpolation procedures implemented in \mathsat,
following \cite{HJMM04}.
In this section, we experimentally evaluate the effectiveness of our new technique.
We will call our implementation of the various algorithms as follows:
\begin{itemize}
\item 
\concpdr is the ``concrete'' IC3 extension for Linear Rational Arithmetic (\lra) as presented in \cite{CGCAV12};
\item 
\conctreepdr is the CFG-based variant of \cite{CGCAV12}, 
  also working only over \lra, 
  and exploiting interpolants whenever possible\footnote{See \cite{CGCAV12} for more details.};
\item 
\abspdr{\T} is IC3 with Implicit Abstraction for an arbitrary theory \T;
\item 
\abstreepdr{\T} is the CFG-based IC3 with Implicit Abstraction for an arbitrary theory \T.
\end{itemize}

All the experiments have been performed on a cluster 
of 64-bit Linux machines with a 2.7 Ghz Intel Xeon X5650 CPU,
with a memory limit set to 3Gb and a time limit of 1200 seconds (unless otherwise specified).
The tools and benchmarks used in the experiments are available at 
\url{https://es.fbk.eu/people/griggio/papers/tacas14-ic3ia.tar.bz2}.

\subsection{Performance Benefits of Implicit Abstraction}
\label{sec:expeval:performance}

\begin{figure}[t!]
  \hspace{-2ex}
  \begin{tabular}{cc@{\hspace{1em}}cc}
    \rotatebox{90}{\hspace{6em}\abspdr{\lra}} &
    \includegraphics[scale=0.55]{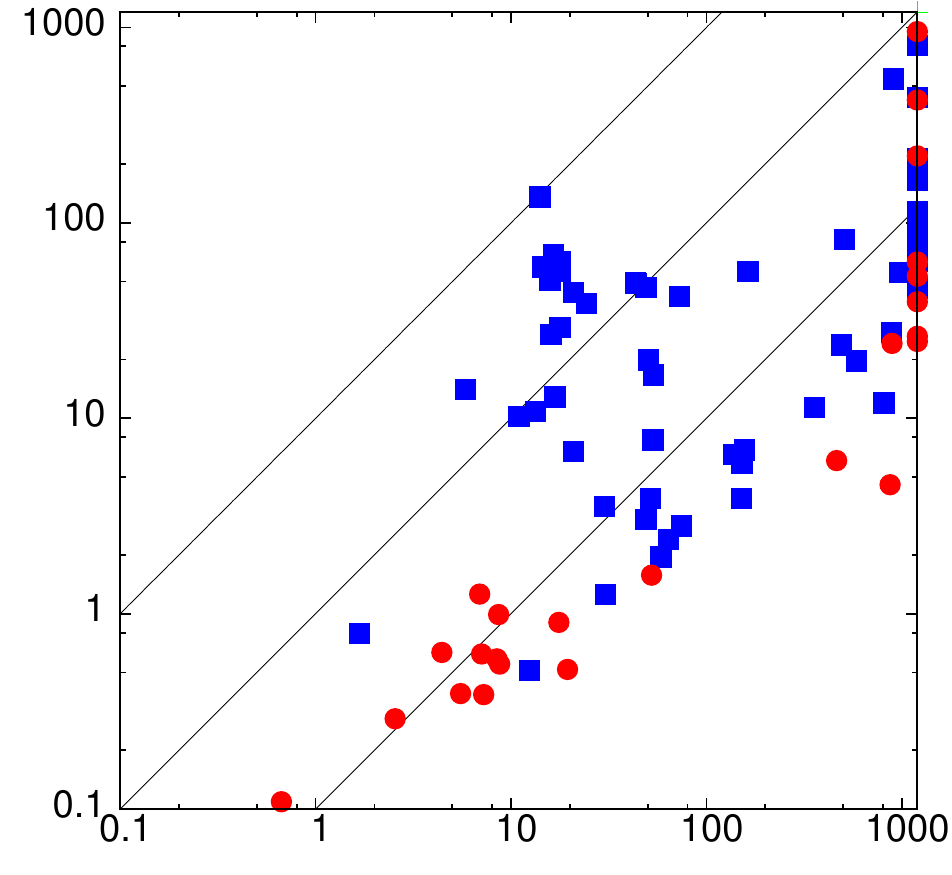} &
    \rotatebox{90}{\hspace{4em}\abstreepdr{\lra}} &
    \includegraphics[scale=0.55]{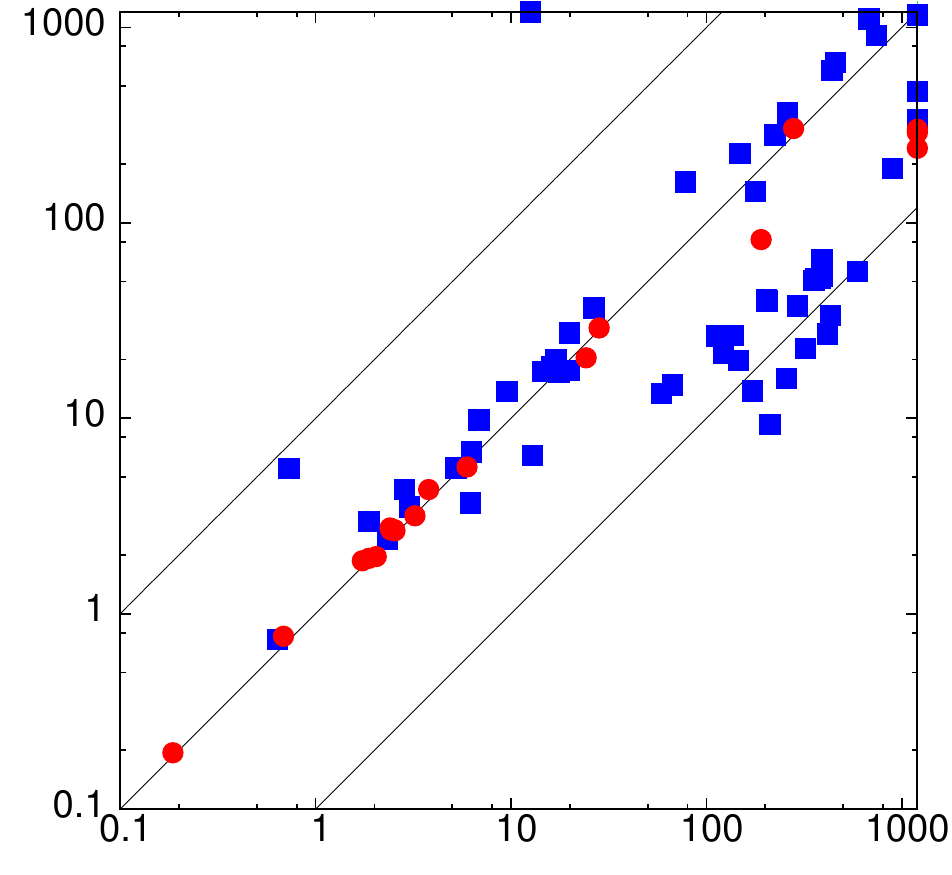} \\
    & \concpdr & & \conctreepdr \\[1em]
    \multicolumn{2}{c}{
    \includegraphics[scale=0.95]{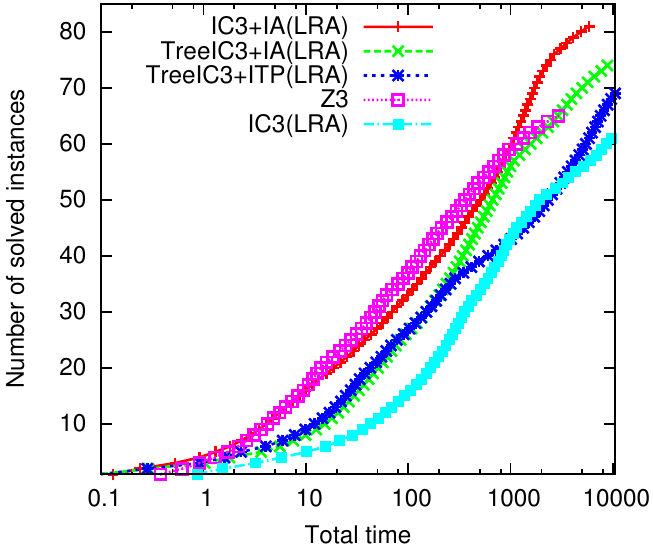}
    } &
    \multicolumn{2}{c}{
      \begin{tabular}[b]{l@{\hspace{1em}}rr}
        & & \\
        \hline
        \multicolumn{2}{l}{\bf Algorithm/Tool \phantom{X}\hfill \# solved} & {\bf Tot time} \\
        \hline
        \abspdr{\lra}      &            82   &       5836     \\
        \abstreepdr{\lra}  &            75   &       8825     \\
        \conctreepdr       &            70   &      10478     \\
        \zthree            &            66   &       2923     \\
        \concpdr           &            62   &       9637     \\
        \hline
        & & \\
        & & \\
        & & \\
        & & \\
        & & \\
        & & \\
        & & \\
      \end{tabular}
    }
    \\
  \end{tabular}
  \caption{Experimental results on \lra benchmarks from \cite{CGCAV12}.  \label{fig:expeval-lra}}
\end{figure}

\sloppypar
In the first part of our experiments,
we evaluate the impact of Implicit Abstraction for the performance of IC3 modulo theories.
In order to do so, we compare \abspdr{\lra} and \abstreepdr{\lra} against \concpdr and \conctreepdr
on the same set of benchmarks used in \cite{CGCAV12},
expressed in the \lra theory.
We also compare both variants against the SMT extension of IC3 for \lra 
presented in \cite{HoderB12} and implemented in the \zthree SMT solver.%
\footnote{We used the Git revision {\tt 3d910028bf} of \zthree.}

The results are reported in Figure~\ref{fig:expeval-lra}.
In the scatter plots at the top, safe instances are shown as blue squares, 
and unsafe ones as red circles.
The plot at the bottom reports the number of solved instances and the total accumulated execution time for each tool.
From the results, we can clearly see that using abstraction has a very significant positive impact on performance.
This is true for both the fully symbolic and the CFG-based IC3,
but it is particularly important in the fully symbolic case:
not only \abspdr{\lra} solves 20 more instances than \concpdr,
but it is also more than one order of magnitude faster in many cases,
and there is no instance that \concpdr can solve but \abspdr{\lra} can't.
In fact, Implicit Abstraction is so effective for these benchmarks that \abspdr{\lra} outperforms also \abstreepdr{\lra}, 
even though \concpdr is significantly less efficient than \conctreepdr.
One of the reasons for the smaller performance gain obtained in the CFG-based algorithm 
might be that \conctreepdr already tries to avoid expensive quantifier elimination operations whenever possible,
by populating the frames with clauses extracted from interpolants,
and falling back to quantifier elimination only when this fails (see \cite{CGCAV12} for details).
Therefore, in many cases \conctreepdr and \abstreepdr{\lra} end up computing very similar sets of clauses.
However, implicit abstraction still helps significantly in many instances, 
and there is only one problem that is solved by \conctreepdr but not by \abstreepdr{\lra}.
Moreover, both abstraction-based algorithms outperform all the other ones, including \zthree.

We also tried a traditional CEGAR approach
based on explicit predicate abstraction, using a bit-level IC3 
as model checking algorithm and the same interpolation procedure of \abspdr{\lra} for refinement.
As we expected, this configuration ran out of time or memory on most of the instances,
and was able to solve only 10 of them.

Finally, we did a comparison with a variant of IC3
specific for timed automata,
\atmoc~\cite{DBLP:conf/formats/KindermannJN12}.
We randomly selected a subset of the properties provided with \atmoc,
ignoring the trivial ones (i.e. properties that are 1-step inductive or 
with a counterexample of length $< 3$). In this case, the comparison
is still preliminary and it should be extend it with more properties
and case studies.
The results are reported in Figure~\ref{fig:expeval-ta}.
\abspdr{\lra} performs very well also in this case, 
solving
100 instances in 772 seconds,
 while \atmoc solved 41 instances in 3953 seconds (\zthree and
\concpdr solved 100 instances in 1535 seconds and 46 instances in 3347
seconds respectively).
\begin{figure}[h!]
  \hspace{-2ex}
  \begin{tabular}{c@{\hspace{3em}}cc}
    \includegraphics[scale=0.9]{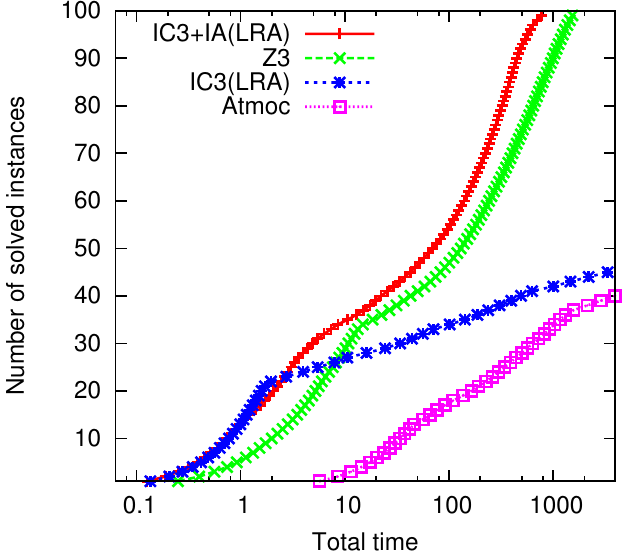} &
    \rotatebox{90}{\hspace{6em}\abspdr{\lra}} &
    \includegraphics[scale=0.55]{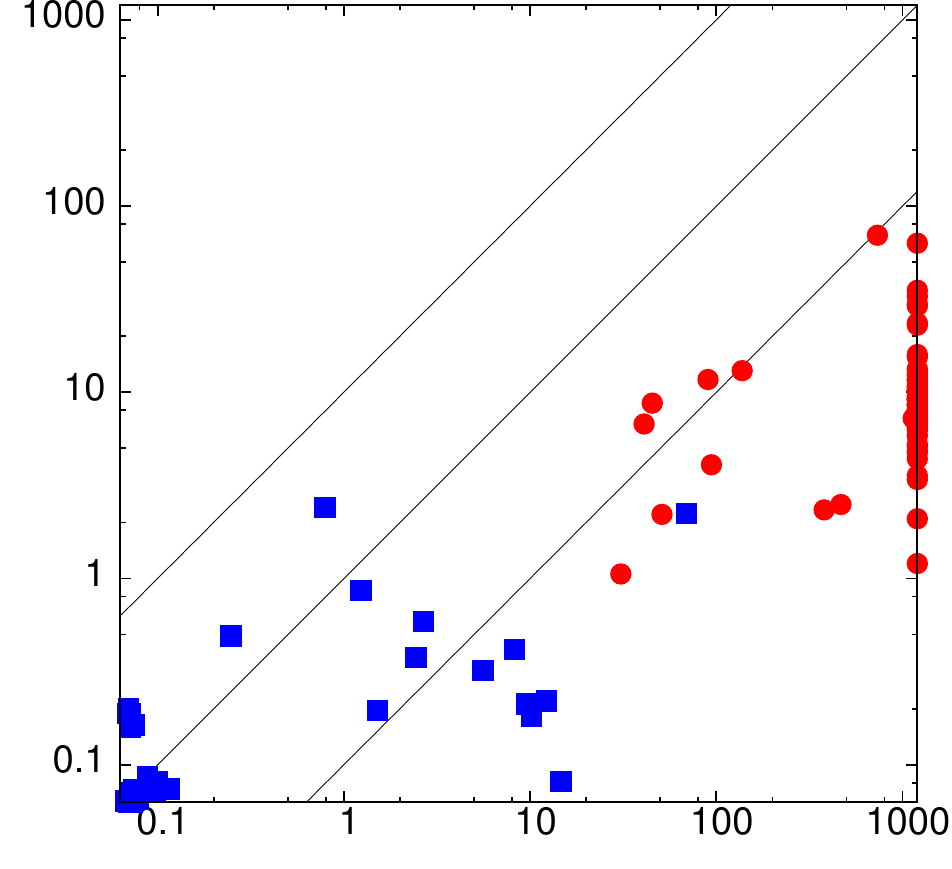} \\
    & & \concpdr
  \end{tabular}
  \caption{Preliminary experimental results on timed automata benchmarks.\label{fig:expeval-ta}}
\end{figure}

\subsubsection{Impact of number of predicates}
\label{sec:expeval:minpreds}

The refinement step may introduce more predicates than
those actually needed to rule out a spurious counterexample
(e.g. the interpolation-based refinement adds all the predicates found
in the interpolant).
In principle, such redundant predicates might significantly hurt performance.
Using the implicit abstraction framework, however,
we can easily implement a procedure that identifies and removes (a subset of)
redundant predicates after each successful refinement step.
Suppose that $\pdria$ finds a spurious
counterexample trace $\pi=(s_0,0); \ldots; (s_k, k)$ with the
set of predicates $\preds$, and that
$\textit{refine}(I,T,\preds,\pi)$ finds a set $\preds_{n}$ of new predicates.
The reduction procedure exploits the encoding of the set of paths of the
abstract system $S_{\preds \cup \preds_n}$ up to $k$ steps,
$\bmc{k}{\preds \cup \preds_n}$.
If $\preds \cup \preds_n$ are sufficient to rule out the spurious
counterexample, $\bmc{k}{\preds \cup \preds_n}$ is
unsatisfiable. We ask the SMT solver to compute the unsatisfiable
core of $\bmc{k}{\preds \cup \preds_n}$, and we keep only the
predicates of $\preds_n$ that appear in the unsatisfiable core.

In order to evaluate the effectiveness of this simple approach,
we compare two versions of \abspdr{\lra} with and without the
reduction procedure.
The results are shown in the scatter plots in Figure~\ref{fig:expeval-predmin},
both in terms of total number of predicates generated (left) and of execution time (right).
Perhaps surprisingly, although the reduction procedure 
is almost always effective in reducing the total number of predicates
\footnote{Sometimes the number of predicate increases. This is not strange, because e.g. it might happen that a predicate that is redundant 
for the current counterexample might become necessary later, and removing it could actually harm in such cases.},
the effects on the execution time are not very big.
Although redundancy removal seems to improve performance
for the more difficult instances, 
overall the two versions of \abspdr{\lra} solve the same number of problems.
However, this shows that the algorithm is much less sensitive to the number of predicates added than approaches based on an explicit computation of the abstract transition relation e.g. via All-SMT,
which often show also in practice (and not just in theory) 
an exponential increase in run time with the addition of new predicates.
\abspdr{\lra} manages to solve problems for which it discovers several
hundreds of predicates, reaching the peak of 800 predicates and solving most of
safe instances with more than a hundred predicates. These numbers are typically way out
of reach for explicit abstraction techniques, which blow up with a few dozen
predicates.
\begin{figure}[h!]
  \hspace{-2ex}
  \begin{tabular}{cc@{\hspace{1em}}cc}
    & Number of predicates & & Execution time \\
    \rotatebox{90}{\hspace{6em}\abspdr{\lra}} &
    \includegraphics[scale=0.55]{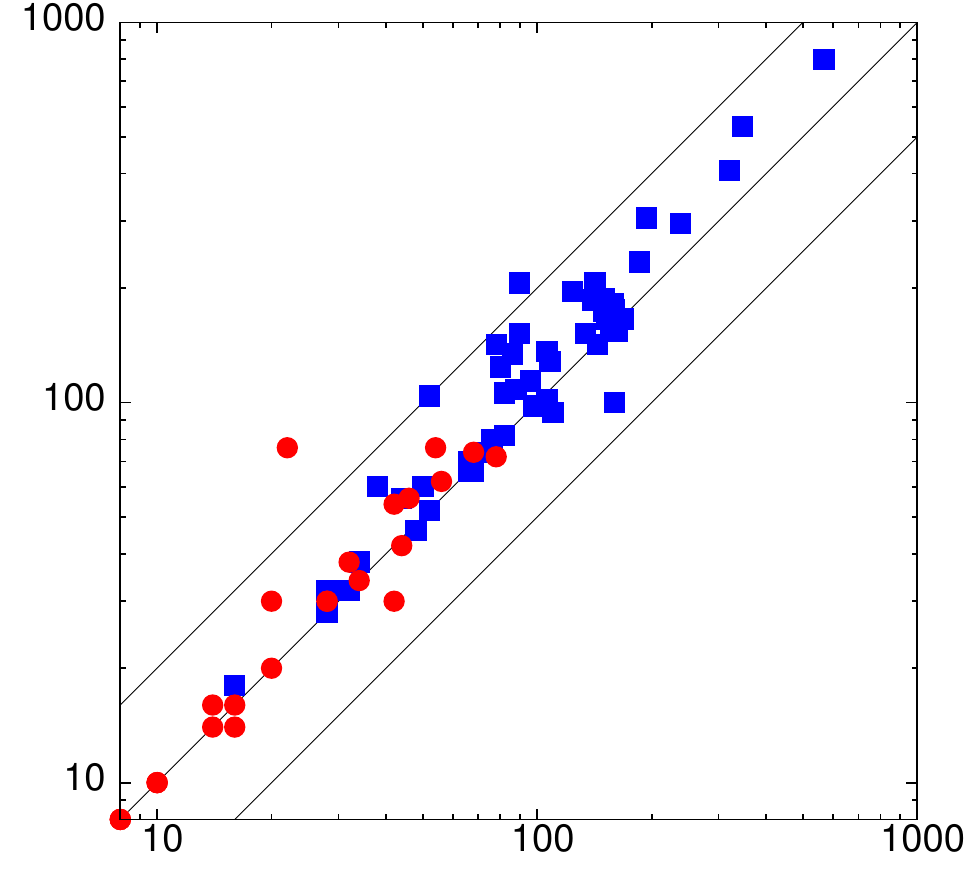} &
    \rotatebox{90}{\hspace{6em}\abspdr{\lra}} &
    \includegraphics[scale=0.55]{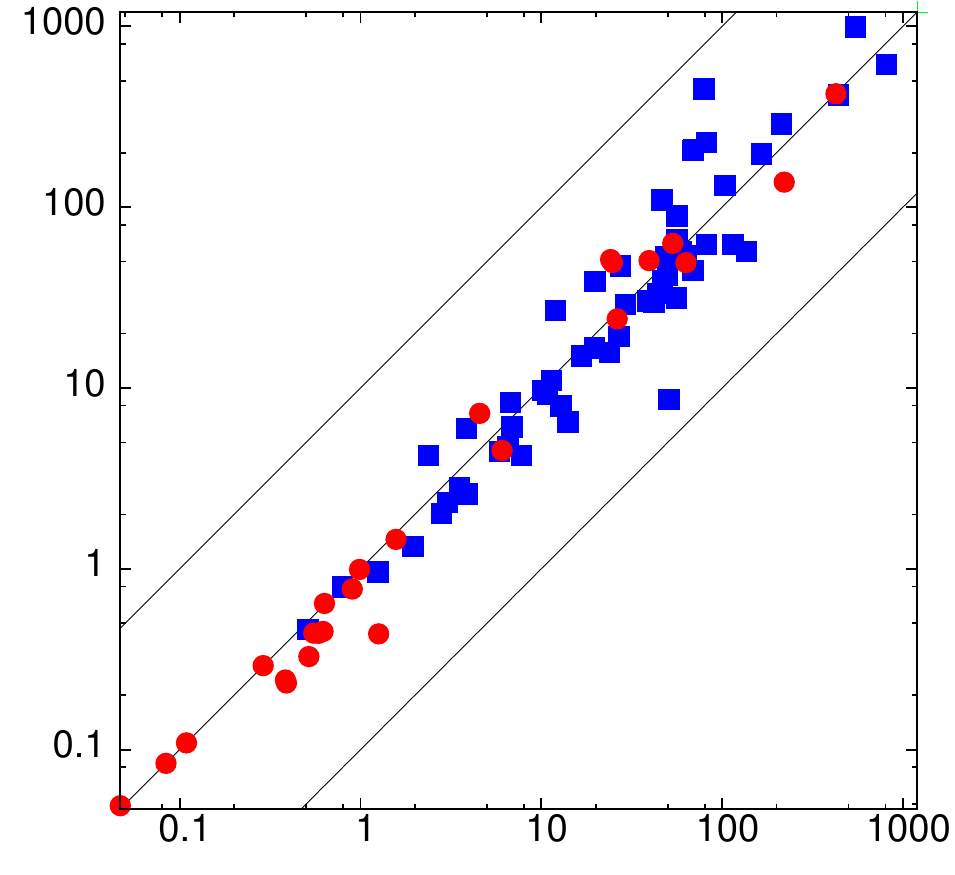} \\
    & \abspdr{\lra} with pred. reduction & & \abspdr{\lra} with pred. reduction
  \end{tabular}
  \caption{Effects of predicate minimization on \abspdr{\lra}.  \label{fig:expeval-predmin}}
\end{figure}


\subsection{Expressiveness Benefits of Implicit Abstraction}
\label{sec:expeval:expressiveness}

In the second part of our experimental analysis,
we evaluate the effectiveness of Implicit Abstraction 
as a way of applying IC3
to systems that are not supported by the methods of \cite{CGCAV12},
by instantiating \abspdr{\T} (and \abstreepdr{\T}) 
over the theories of Linear Integer Arithmetic~(\lia) 
and of fixed-size bit-vectors~(\bv).

\vspace{-1em}
\subsubsection{IC3 for \bv}
\label{sec:expeval:bv}

\begin{figure}[t!]
  \vspace{-1em}
  \hspace{-2ex}
  \begin{tabular}{cc}
    \includegraphics[scale=0.95]{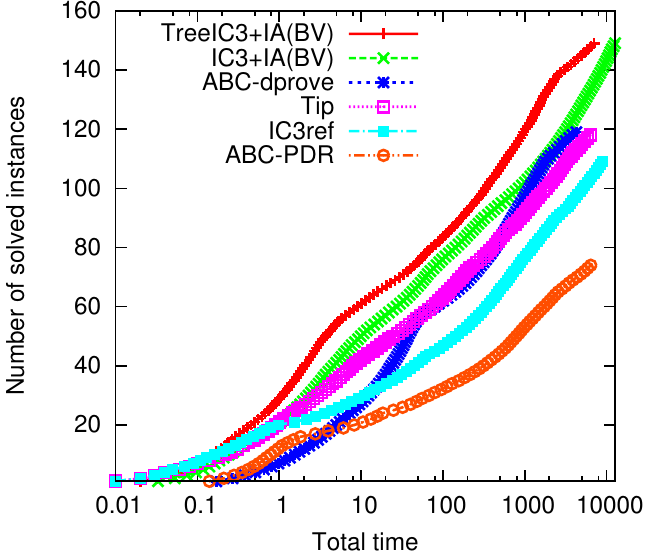}
    &
    \begin{tabular}[b]{l@{\hspace{1em}}rr}
      & & \\
      \hline
      \multicolumn{2}{l}{\bf Algorithm/Tool \phantom{X}\hfill \# solved} & {\bf Tot time} \\
      \hline
      \abstreepdr{\bv}   &           150   &       7056     \\
      \abspdr{\bv}       &           150   &      12753     \\
      \abcdprove         &           120   &       4298     \\
      \tip               &           119   &       6361     \\
      \icthree           &           110   &       9041     \\
      \abcpdr            &            75   &       6447     \\
      \hline
      & & \\
      & & \\
      & & \\
      & & \\
      & & \\
      & & \\
    \end{tabular}
    \\
  \end{tabular}
  \caption{Experimental results on \bv benchmarks from software verification. \label{fig:expeval-bv}}
\end{figure}

For evaluating the performance of \abspdr{\bv} and \abstreepdr{\bv},
we have collected over 200 benchmark instances from the domain of software verification. 
More specifically, the benchmark set consists of:
\begin{itemize}
\item 
all the benchmarks used in \S\ref{sec:expeval:performance}, but using \bv instead of \lra as background theory;
\item 
the instances of the {\tt bitvector} set of the Software Verification Competition SV-COMP~\cite{swmcc};
\item 
the instances from the test suite of InvGen~\cite{invgen}, a subset of which was used also in \cite{DBLP:conf/date/WelpK13}.
\end{itemize}

We have compared the performance of our tools with various implementations of the Boolean IC3 algorithm,
run on the translations of the benchmarks to the bit-level Aiger format: 
the PDR implementation in the ABC model checker (\abcpdr)~\cite{fmcad-een},
\tip~\cite{tip},
and \icthree~\cite{ic3ref}, the new implementation of the original IC3 algorithm as described in \cite{bradley}.
Finally, we have also compared with the \textsc{dprove} algorithm of ABC (\abcdprove),
which combines various different techniques for bit-level verification, including IC3.%
\footnote{We used ABC version {\tt 374286e9c7bc}, \tip{} {\tt 4ef103d81e} and \icthree{} {\tt 8670762eaf}.
}
We also tried \zthree, but it ran out of memory on most instances.
It seems that \zthree uses a Datalog-based engine for \bv, rather than PDR.

The results of the evaluation on \bv are reported in Figure~\ref{fig:expeval-bv}.
As we can see, both \abspdr{\bv} and \abstreepdr{\bv} outperform the bit-level IC3 implementations.
In this case, the CFG-based algorithm performs slightly better than the fully-symbolic one,
although they both solve the same number of instances.

\vspace{-1em}
\subsubsection{IC3 for \lia}
\label{sec:expeval:bv}

\begin{figure}[t!]
  \vspace{-1em}
  \hspace{-2ex}
  \begin{tabular}{cc}
    \includegraphics[scale=0.95]{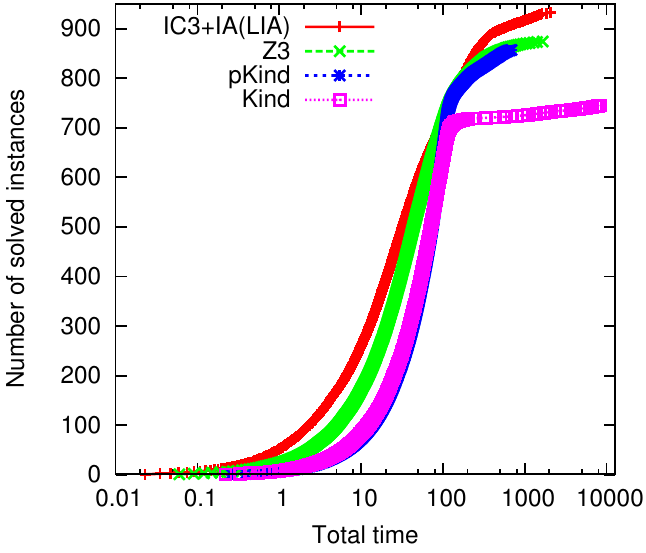}
    &
    \begin{tabular}[b]{l@{\hspace{1em}}rr}
      & & \\
      \hline
      \multicolumn{2}{l}{\bf Algorithm/Tool \phantom{X}\hfill \# solved} & {\bf Tot time} \\
      \hline
      \abspdr{\lia}      &           933   &       2064     \\
      \zthree            &           875   &       1654     \\
      \pkind             &           859   &        720     \\
      \kind              &           746   &       8493     \\
      \hline
      & & \\
      & & \\
      & & \\
      & & \\
      & & \\
      & & \\
      & & \\
    \end{tabular}
    \\
  \end{tabular}
  \caption{Experimental results on \lia benchmarks from Lustre programs \cite{kind-fmcad08}.  \label{fig:expeval-lia}}
\end{figure}

For our experiments on the \lia theory,
we have generated benchmarks using the Lustre programs available from the webpage
of the \kind model checker for Lustre~\cite{kind-fmcad08}.
Since such programs do not have an explicit CFG, we have only evaluated \abspdr{\lia},
by comparing it with \zthree
and with the latest versions of \kind as well as its parallel version \pkind~\cite{pkind}.%
\footnote{We used version {\tt 1.8.6c} of \kind and \pkind, 
which is not publicly available but was provided to us by the \kind authors.
\pkind differs from \kind because it exploits multi-core machines
and complements k-Induction with an automatic invariant generation procedure.
}

The results are summarized in Figure~\ref{fig:expeval-lia}.
Also in this case, \abspdr{\lia} outperforms the other systems.

\section{Conclusion}
\label{sec:conclusion}

In this paper we have presented \pdria, a new approach to the
verification of infinite state transition systems, based on an
extension of IC3 with implicit predicate abstraction.

The distinguishing feature of our technique is that \pdr 
works in an abstract state space, since the counterexamples to
induction and the relative inductive clauses are expressed
with the abstraction predicates. This is enabled by the use of implicit
abstraction to check (abstract) relative induction.
Moreover, the refinement in our procedure is fully incremental, 
allowing to keep all the clauses found in the previous iterations.

The approach has two key advantages. First, it is very general: the
implementations for the theories of LRA, BV, and LIA have been
obtained with relatively little effort. Second, it is extremely
effective, being able to efficiently deal with large numbers of
predicates.
Both advantages are confirmed by the experimental results, obtained on
a wide set of benchmarks, also in comparison against dedicated
verification engines.

In the future, we plan to apply the approach to other theories
(e.g. arrays, non-linear arithmetic) investigating other forms of
predicate discovery, and to extend the technique to
liveness properties.

\bibliographystyle{splncs03} 
\bibliography{main}

\end{document}